\documentclass[pra,showpacs,twocolumn,superscriptaddress,nofootinbib]{revtex4-1}

\usepackage{lipsum}

\usepackage{mathtools}
\usepackage{float}
\usepackage{graphicx}
\usepackage{epsfig}
\usepackage{subfig}

\usepackage{dcolumn}
\usepackage{amsmath,amssymb,amsthm}
\usepackage{bm}
\usepackage{verbatim}
\usepackage{caption}
\usepackage{color}

\newtheorem{theorem}{Theorem}
\newtheorem{definition}{Definition}

\newcommand{\bra}[1]{\langle #1|}
\newcommand{\ket}[1]{|#1\rangle}

\newcommand{\ketbra}[2]{| #1 \rangle \langle #2 |}

\newcommand{\M}[1]{\mathcal{#1}}

\newcommand{\ad}[1]{a^{\dagger}_{#1}}

\newcommand{\vacket}{\ket{0}}
\newcommand{\vacbra}{\bra{0}}

\newcommand{\N}[1]{\Vert #1 \Vert}
\newcommand{\id}[1]{\mathbb{I}}
\newcommand{\Tr}[1]{\text{Tr}}

\usepackage{ulem}

\begin{document}
\title{ Completely Positive Maps for  Reduced States of Indistinguishable Particles } 
\author{Leonardo da Silva Souza}
\email{leonardosilvsouza@gmail.com}
\affiliation{Departamento de F\'{\i}sica - ICEx - Universidade Federal de Minas Gerais,
Av. Pres. Ant\^onio Carlos 6627 - Belo Horizonte - MG - Brazil - 31270-901.}
\author{Tiago Debarba} \affiliation{Universidade Tecnol\'ogica Federal do Paran\'a (UTFPR), Campus Corn\'elio Proc\'opio, 
Avenida Alberto Carazzai 1640, Corn\'elio Proc\'opio, Paran\'a  86300-000, Brazil}
\author{Diego L. Braga Ferreira}\affiliation{Departamento de F\'{\i}sica - ICEx - Universidade Federal de Minas Gerais,
Av. Pres. Ant\^onio Carlos 6627 - Belo Horizonte - MG - Brazil - 31270-901.}
\author{Fernando Iemini}\affiliation{Abdus Salam ICTP, Strada Costiera 11, I-34151 Trieste, Italy}
\author{Reinaldo O. Vianna}\affiliation{Departamento de F\'{\i}sica - ICEx - Universidade Federal de Minas Gerais,
Av. Pres. Ant\^onio Carlos 6627 - Belo Horizonte - MG - Brazil - 31270-901.}

\date{\today}
\begin{abstract}
We introduce  a framework for the construction of  completely positive maps for subsystems of  indistinguishable fermionic particles. 
In this scenario, the initial global state is always correlated, and  it is not possible to tell system and environment apart. 
Nonetheless, a reduced map in the operator sum representation is possible for some sets of states where the only non-classical correlation present is exchange.
\end{abstract}
\maketitle

\section{Introduction}

\begin{figure*}[t]
\centering
\includegraphics[scale=0.3]{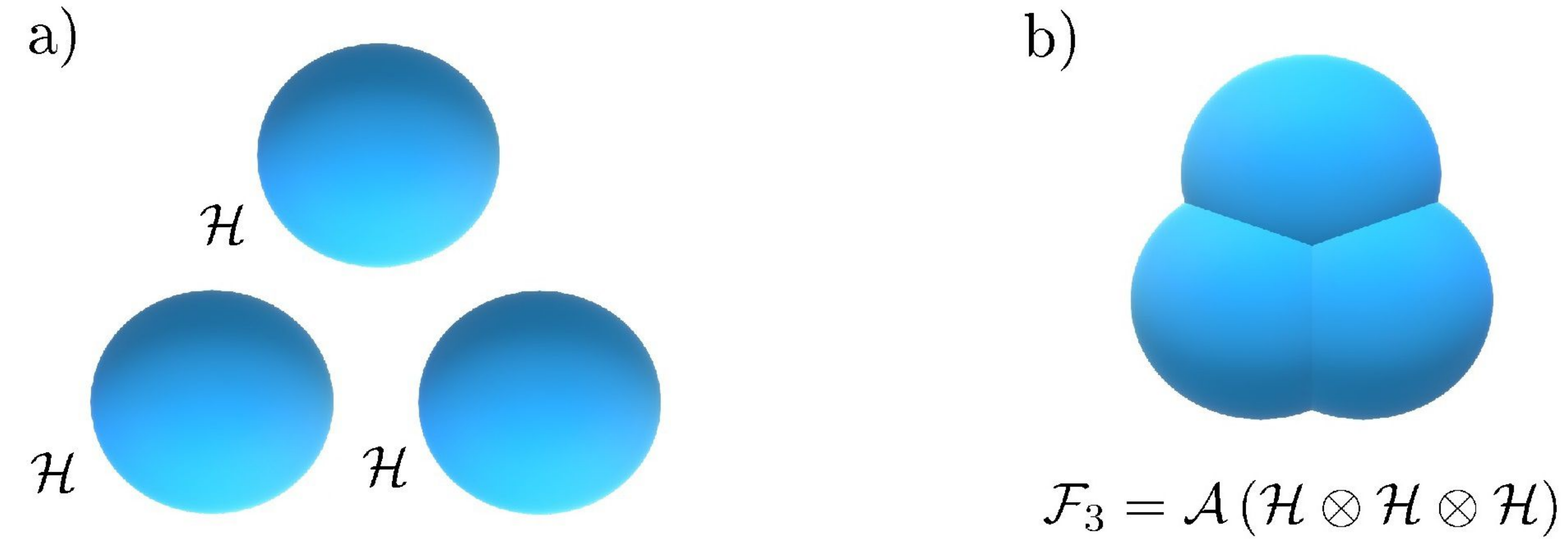}
\caption{Pictorial view of Hilbert space  with (a) tensor product structure ($\M{H}\otimes\M{H}\otimes\M{H}$),  and (b) antisymmetric space without tensor product structure, where the 
particle states overlap.  
A partial trace over a subsystem in the antisymmetric space has information about the whole system, since the particles are indistinguishable. }
\end{figure*}
The characterization of the dynamics of a system that may be correlated with other systems has been subject of investigation in several areas,
 varying from quantum information processing to condensed matter physics \cite{Breuer,Nielsen}. 
A closed system evolves unitarily according to the Schr\"odinger equation. On the other hand, the dynamics of a subsystem is not 
necessarily unitary, and the theory of open quantum systems provides the mathematical framework to treat it. In this context, we
speak of system and environment, and say that the system, which is just a part of the whole, is open. 
If system and environment start in a uncorrelated global state (factorable), then the dynamics is guaranteed to be completely positive (CP).
However, if the system is initially correlated with the environment, the map associated with the dynamics of the system may not be
 completely positive or, as we will see, is valid only for a subset of the state space. 
 In recent years, more attention has been given to the construction of reduced dynamical maps with different initial conditions \cite{SB,HKO, RMKSS,SL,BDMRR,VA}, mainly motivated by discussions between Pechukas and Alicki \cite{PP,RA, PP2}.
  Pechukas introduced  the idea of  `assignment map' ($\Phi$), which characterizes initial system-environment states ($\Phi\rho_S=\rho_{SE}$) for open quantum systems, and showed that imposing  three `natural'  conditions, namely: (linearity) $\Phi$ preserves mixtures;  (consistency) it is consistent, in the sense that $\rho_S=Tr_E(\Phi\rho_S)$;  (positivity) and $\Phi\rho_S$ is positive for all positive $\rho_S$; this implies  the initial state of the system and environment is factorable ($\Phi\rho_S=\rho_S\otimes\rho_E$). 
  To deal with the problem of characterizing reduced dynamics of initial correlated systems,  Pechukas \cite{PP, PP2} suggested to giving up positivity. On the other hand, Aliciki  \cite{RA} argued to  either giving up consistency or linearity.  In the end, the conclusion is that, one way or the other, the domain of validity of the assignment map must be restricted. 
Afterwards,   Stelmachovic et al. \cite{SB}  studied   the influence  of initial correlations between  system and environment in the dynamics of the system, making clear that taking into account such correlations is paramount to the correct description of the evolution.  They showed an instructive example with two qubits (one for the system, one for the environment), evolving under a C-NOT gate: both a maximally entangled
state and a maximally mixed global state have the same one-qubit local maximally mixed states, but the evolution is radically different.
In a comment to \cite{SB}, Salgado et al. \cite{Salgado}  showed for two qubits   that, whatever the initial 
correlations, the system dynamics has the  Kraus representation form,  and is consequently completely 
positive, whenever the global  dynamics is locally unitary. This was then 
proved for bipartite global systems of arbitrary dimension by Hayashi et al. \cite{HKO}.
Later on many authors  worked out    sets of classicaly \cite{RMKSS,SL}   or quantum  \cite{BDMRR,VA} correlated initial global states 
 that guarantee complete positivity of the reduced dynamics.  The subject has recently regained impetus, with many interesting discussions 
  \cite{VA, PRL2016, PRA2013, QIP2016, LAA2015, PRL2015, arxiv2018}.

	In this work we are interested  in the construction of the reduced dynamical map in the case  of  
  systems of $N$ indistinguishable particles, in particular fermions,  which are always correlated,
   and for which an usual tensor product structure between  `system' and `environment' is absent.
   The subtle notion of  quantum correlations of indistinguishable particles has been investigated by many authors, 
    with introduction of seminal ideas, as entanglement of modes \cite{Zanardi}, or entanglement 
    of particles \cite{ESBL,WV,Balachandran2013,iemini13,Iemini13B,Rossignoli1,Rossignoli2}. 
    Our own group has scrutinized the concept of entanglement of particles \cite{iemini13, Iemini13B}, and made interesting applications \cite{Iemini15}.
    More recently,  the concept of `quantumness of correlations' of indistinguishable particles was explored  by 
    Iemini et al. \cite{IDV}, and Debarba et al. \cite{DIV}. 
    It is well established that the  exchange correlations generated by mere antisymmetrization of the state, due to indistinguishability of their fermions,  does not result in 
    entanglement or,
  more generally,   in  quantumness  \cite{IDV,DIV}.  To the best of our knowledge,  the role of initial exchange correlations in the reduced dynamics is still unexplored. 
  We  propose a framework to construct  completely positive maps representing the dynamics of a single particle reduced state.     
  
This paper is organized as follows. In Sec.II we briefly discuss particle correlation in the antisymmetric subspace. In Sec.III we identify a class of initial  
global states that  give rise to completely positive reduced dynamics. In Sec.IV 
we illustrate the formalism with an example of two fermions under
 a quadratic Hamiltonian. Conclusions are presented in Sec.V.

\section{Correlations in the antisymmetric subspace }

Composed distinguishable quantum systems are described by density operators over a composition of Hilbert spaces of individual subsystems, by means of  the tensor product: 
\begin{equation}\label{hspace}
\rho_{1\cdots N} : \M{H}_N \longrightarrow  \M{H}_N,
\end{equation}
where $\M{H}_N= \M{H}_1^{L_1} \otimes\cdots\otimes\M{H}_N^{L_N} $,  $N$ is the number of 
subsystems,  $L_i$ is the dimension of $i$'th subsystem, 
 and  $\rho_{1\cdots N}\in\M{D}(\M{H}_1^{L_1} \otimes\cdots\otimes\M{H}_N^{L_N})$, with $\M{D}$ the set of density operators (positive semidefinite and trace-one operators). 
In these systems, the tensor product structure between the subsystems plays an important role to the characterization of correlations as entanglement \cite{horodeckireview} and quantumness  \cite{zurek01,reviewdiscord}.  
However, the state space  of $N$ indistinguishable fermions is described by the antisymmetrized composed Hilbert space (Fig. 1): 
\begin{equation}
\M{F}_N^L = \M{A}(\M{H}_1^{L} \otimes\cdots\otimes\M{H}_N^{L} ), 
\end{equation}
where $N$ is the number of fermions and $L$ is the number of accessible modes. 
Note that this space does not support a tensor product structure and have a more suitable description in the second quantization formalism.
Therefore a basis in this subspace can be constructed out  of  fermionic operators $\{ a_k \}_{k=1}^L$, satisfying the usual  anti-commutation relations:
\begin{equation}\label{anticom}
\{a_l,\ad{k} \} = \delta_{k,l}, \qquad \{a_k,a_l\}=\{\ad{k},\ad{l}\} =0,
\end{equation}
 where $a_k$ and $\ad{k}$ are annihilation and creation operators for the $k$'th mode, respectively. A single 
 particle orthonormal basis is formed by the set of states $\{ \ad{k}\vacket\}_{k=1}^L$, 
 where $\vacket$ represents the vacuum. 

As mentioned in the Introduction, the correlation of indistinguishable particles, mostly entanglement, was study by many groups \cite{ESBL,WV,Balachandran2013,iemini13,Iemini13B,Rossignoli1,Rossignoli2},  giving rise to many definitions that agree with each other in the fermionic case, in the sense that the set of unentangled states can be written as a convex sum of Slater determinants. More generally, with studies in quantumness \cite{IDV,DIV}, we can define states where the only non-classical correlation present is exchange, which leads to the following definition: 

\begin{definition}A fermionic state $\omega\in\M{D}\left(\M{F}_N^L\right)$ has no quantumness of correlation if it can be decomposed as a convex combination of orthogonal Slater determinants, namely,
\begin{equation}\label{unc.qstate} 
\omega = \sum_{\vec{k}}  p(\vec{k})  \ad{\vec{k}} \vacket\vacbra a_{\vec{k}},
\end{equation}
 where $\vec{k} = (k_1,\ldots, k_N)$ is an $N$-tuple denoting the modes occupied by the fermions,
 with $k_i = 1,...,L$, $p(\vec{k})$ is a probability distributions and
  $\ad{\vec{k}} \vacket \equiv \ad{k_1}\cdots\ad{k_N}\vacket$.
\end{definition}

As we are interested in exploring the role of initial exchange correlations in the reduced dynamics of fermionic systems, 
we will choose the initial global fermionic state in the set with no quantumnes, according to Definition 1.


\section{ Dynamical Maps for Reduced States of Fermionic Systems}
In this section we introduce  the formalism to describe the dynamics of a single fermion in a closed system of $N$ fermions. 
More precisely, given a system of $N$ indistinguishable fermions in the state 
 $\rho(0)$, evolving under the unitary  $U_t$, which  preserves the total number of particles, we want to obtain the dynamical map $\Phi_t$, 
 which evolves   the one-particle reduced state 
$\rho_{r}=Tr_{N-1}\left(\rho(0)\right)$, see Fig. 2.
 Since the fermionic states are restricted to the antisymmetric 
sector of the Hilbert space, it is not possible to start with initial states in the tensor product form.
As discussed in the Introduction,  one way to 
deal with the problem of obtaining completely positive maps, characterizing the dynamics of states initially correlated with an external  system, is to restrict the 
domain of the map. Using the fact that the Kraus representation  assures completely positivity \cite{Breuer,Nielsen}, we will 
show that for some sets of initial states with no quantumness of correlations, we can construct  completely positive  maps for the reduced state.  
\begin{figure}[t]\label{diagram}
\includegraphics[scale=0.28]{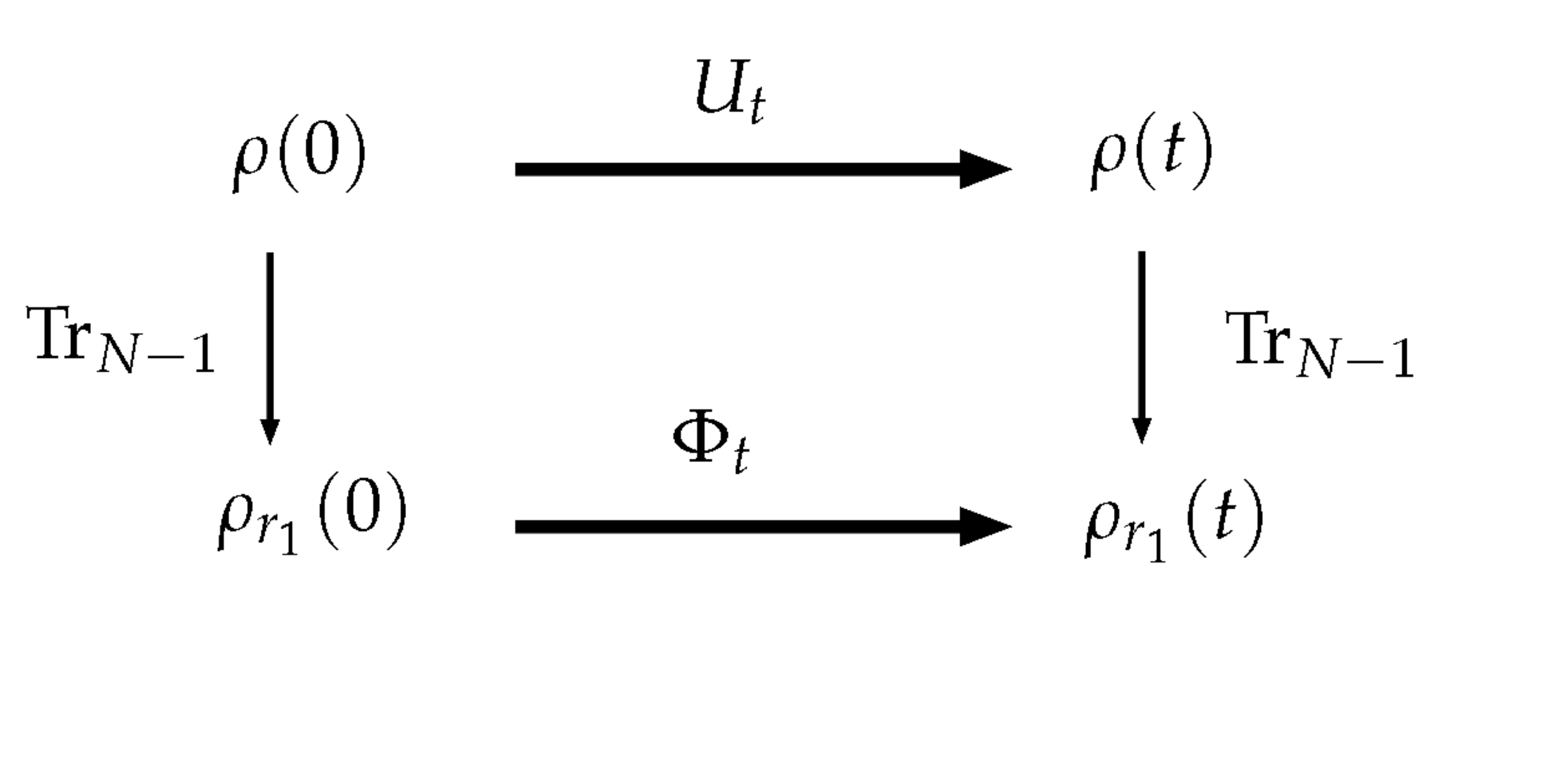}
\caption{Schematic diagram characterizing the dynamics of indistinguishable fermions.  Suppose an initial  $N$-fermion state $\rho(0)$  evolving under the unitary $U_t$. 
The reduced one-fermion state $\rho_{r}(0)=Tr_{N-1}(\rho(0))$  evolves under the dynamical map $\Phi_t$.}
\end{figure}

The construction of the single-fermion dynamical map, in the simplest scenario of a closed system of two fermions  initially in a pure state, $\rho(0) = \ketbra{\psi}{\psi}$,  gives us a good grasp on the general features of the formalism, and includes all the technical problems of the general case. The generalisation to $N$ fermions mixed states is straightforward and performed
in Appendix B.

Let us consider a set of states   in  the antisymmetric space of 2 fermions and $L+1$ modes, that can be written in a given basis
 of  Slater determinants as: 
\begin{eqnarray}\label{rho}
\mathcal{S}_{2}^{\mu,pure} \equiv \left\{ a_{\mu}^{\dagger}a_{k}^{\dagger}\ketbra{0}{0}a_{k}a_{\mu} \right\}_{k=0}^L ,
\end{eqnarray}
where $\mu$ is fixed mode.
Note that $\mu$ labels a reference mode, and different  values of $\mu$ lead to 
distinct sets.

Let us calculate the one-particle reduced state by tracing out  one 
fermion from Eq.(\ref{rho}). Assuming that $\{f_{k}^{\dagger}\}_{k=0}^{L}$
is an orthonormal basis of fermionic creation operators for the space of a single fermion $(\mathcal{F}_{1}^{L+1})$, thus $f_{k}^{\dagger}=\sum_{l}V_{kl}a_{l}^{\dagger}$, $V$ is a unitary matrix of dimension $L+1$.
The partial  trace over  one particle is given by $\rho_{r}=\frac{1}{2}\sum_{k=0}^{L}{f_{k}\rho f_{k}^{\dagger}}$.
  The explicit calculation of   the matrix element $(\rho_{r})_{i,j}$  goes as follows:
   \begin{eqnarray}\label{part.trac}
(\rho_{r})_{i,j} & =\bra{0}f_{j}\left(\frac{1}{2}\sum_{k=0}^{L}{f_{k}\rho f_{k}^{\dagger}}\right)f_{i}^{\dagger}\ket{0}\nonumber\\
                           & =\frac{1}{2}\sum_{k=0}^{L}{\bra{0}f_{k}\left(f_{j}\rho f_{i}^{\dagger}\right)f_{k}^{\dagger}}\ket{0}\nonumber\\
                          & =\frac{1}{2}Tr_1(f_{i}^{\dagger}f_{j}\rho),
\end{eqnarray}
where we used  the  fermionic anti-commutation  relations and the cyclicality of the trace.
 Now we can write the set of single-fermion reduced states of Eq.(\ref{rho}): 
\begin{eqnarray}\label{dom.2pur}
\mathcal{S}_{r(2)}^{\mu,pure} &=& Tr_1\left( \mathcal{S}_2^{\mu,pure} \right) \nonumber\\
&=& \left\{\frac{1}{2} a_{k}^{\dagger}\ketbra{0}{0}a_{k}+\frac{1}{2}a_{\mu}^{\dagger}\ketbra{0}{0}a_{\mu} \right\}_{k=0}^L, \label{rho0}
\end{eqnarray} 
with $\mu$ a fixed mode.
Assuming  the dynamics of  $\rho(0)\in\mathcal{S}_{2}^{\mu,pure}$  is given by the 
unitary  $U_t$, we can define a CP map $\Phi_t^{\mu}$  for the dynamics
 of the single-fermion reduced state 
 $\rho_{r}(0)\in\mathcal{S}_{r(2)}^{\mu,pure}$, \textit{i.e.}, a CP map
  $\Phi_t^\mu:\mathcal{S}_{r(2)}^{\mu,pure}\mapsto\mathcal{F}_{1}^{L+1}$ as follows:

\begin{definition}\label{defk1}

A dynamical map $\Phi_t^{\mu}$  for the single-fermion reduced state $\rho_{r}(0)\in\mathcal{S}_{r(2)}^{\mu,pure}$, of a 2-fermion pure state  initially with no quantumness of correlations, 
 $\rho(0)\in\mathcal{S}_{2}^{\mu,pure}$, evolving under the global unitary $U_t$,
has the  operator sum  representation 
$\Phi_t^{\mu}[\rho_{r}]=\sum_{j=0}^L K_{j}^{\mu}\rho_{r}K_{j}^{\dagger\mu}$, with  the Kraus operators 
\begin{equation}\label{kraus2p}
K_{l}^{\mu}=f_{l}U_ta_{\mu}^{\dagger}.
\end{equation}
\end{definition}

\begin{proof} 
If the  2-fermion state  evolves according to  $\rho(t)= U_{t}\rho(0)U_{t}^{\dagger}$,  the reduced density matrix is:
\begin{eqnarray}\label{rhort}
\rho_{r}(t)&=&Tr_{1}(U_t a_{\mu}^{\dagger}a_{k}^{\dagger}\ketbra{0}{0}a_{k}a_{\mu}U_{t}^{\dagger})\nonumber\\
&=&\sum_{l=0}^{L}{f_{l} U_{t}a_{\mu}^{\dagger}\left(\frac{1}{2}a_{k}^{\dagger}\ketbra{0}{0}a_{k}\right)a_{\mu}U_{t}^{\dagger}f_{l}^{\dagger}},
\end{eqnarray}
where in the last equation we used the definition of fermionic partial trace (Eq.(\ref{part.trac})) and the anti-commutation relations. 
Using the fact that we cannot create more than one fermion in the same mode (Pauli exclusion principle), we can add a second null term  in  Eq.(\ref{rhort}), 
in order to recover the reduced state in the form of Eq.(9),
\begin{eqnarray}
\rho_{r}(t)&=&\sum_{l=0}^{L}{f_{l} U_{t}a_{\mu}^{\dagger}\left(\frac{1}{2}a_{k}^{\dagger}\ketbra{0}{0}a_{k}\right)a_{\mu}U_{t}^{\dagger}f_{l}^{\dagger}}\nonumber\\
&&+\sum_{l=0}^{L}{f_{l} U_{t}a_{\mu}^{\dagger}\left(\frac{1}{2}a_{\mu}^{\dagger}\ketbra{0}{0}a_{\mu}\right)a_{\mu}U_{t}^{\dagger}f_{l}^{\dagger}},
\end{eqnarray}
which can be written as,
\begin{eqnarray}
\rho_{r}(t)&=&\sum_{l=0}^{L}{f_{l} U_{t}a_{\mu}^{\dagger}\left(\rho_{r}(0)\right)a_{\mu}U_{t}^{\dagger}f_{l}^{\dagger}}\nonumber\\
&=&\sum_{l=0}^{L}K_{l}^{\mu}\rho_{r}(0)K_{l}^{\mu\dagger},
\end{eqnarray}
with $  K_{l}^{\mu}=f_{l}U_ta_{\mu}^{\dagger}$.
\end{proof}

 Due to the restriction of the map domain (Eq(\ref{dom.2pur})), the relation between Kraus operators and trace preservation can  be written  as, 
 \begin{equation}
  \sum_{l}{K_{l}^{\mu\dagger}K_{l}^{\mu}}=\rm{diag}\left(\lambda_0,\lambda_{1},...,\lambda_{L}\right),
 \end{equation}
 with $\lambda_{i\neq\mu}=2$ and $\lambda_{\mu}=0$, since
\begin{eqnarray}
Tr\left(\rho_{r}(t)\right)&=&Tr\left[\rm{diag} \left(\lambda_0,\lambda_{1},...,\lambda_{L}\right)\rho_{r}(0)\right]=1,
\end{eqnarray}
where $\rho_{r}(0)\in\mathcal{S}_{r(2)}^{\mu,pure}$. This can be checked by computing the matrix elements of $\sum_{l}{K_{l}^{\mu\dagger}K_{l}^\mu}$, in the basis 
$\{a_{k}^{\dagger}\ket{0}\}_{k=0}^{L}$, namely:
\begin{eqnarray}
\sum_{l=0}^{L}\left({K_{l}^{\mu\dagger}K_{l}^{\mu}}\right)_{i,j}&=&
\bra{0}a_{i} \sum_{l=0}^{L} K_{l}^{\mu\dagger} \left(\sum_{k=0}^{L} a_{k}^{\dagger}\ketbra{0}{0} a_{k} \right)K_{l}^{\mu} a_{j}^{\dagger}\ket{0} \nonumber\\
&=&\bra{0}a_{i}a_{\mu}U_{t}^{\dagger}\left(\sum_{k,l}f_{l}^{\dagger}a_{k}^{\dagger}\ketbra{0}{0}a_{k}f_{l}\right)\nonumber\\
&&\times U_{t}a_{\mu}^{\dagger}a_{j}^{\dagger}\ket{0},\nonumber\\
\end{eqnarray}
where we used in the first line the identity $\sum_{k} a_{k}^{\dagger}\ketbra{0}{0} a_{k} = \mathbb{I}_{\mathcal{F}_{1}^{L+1}}$.

Since $\{a_i\}$ and $\{f_i\}$ are both orthonormal bases,
there exists a unitary $V$, of dimension $L+1$, which performs the  
single particle transformation $f_{l}^{\dagger}=\sum_{m}V_{lm}a_{m}^{\dagger}$, we can simplify the term 
\begin{align}
&\left(\sum_{k,l}f_{l}^{\dagger}a_{k}^{\dagger}\ketbra{0}{0}a_{k}f_{l}\right)=\nonumber\\
&=\left(\sum_{k,l,m,n}V_{m,l}V_{n,l}^{*}a_{m}^{\dagger}a_{k}^{\dagger}\ketbra{0}{0}a_{k}a_{n}\right)\nonumber\\
&=\left(\sum_{k,m}a_{m}^{\dagger}a_{k}^{\dagger}\ketbra{0}{0}a_{k}a_{m}\right)=2\mathbb{I}_{\mathcal{F}_{2}^{L+1}},
\end{align}
therefore,we have:
\begin{eqnarray}
\sum_{l=0}^{L}\left({K_{l}^{\mu\dagger}K_{l}^{\mu}}\right)_{i,j}&=&2\bra{0}a_{i}a_{\mu}a_{\mu}^{\dagger}a_{j}^{\dagger}\ket{0}\nonumber\\\nonumber\\
&=&\left\{\begin{array}{l}
		2, \quad  \text{if } i= j,\,i\neq \mu,\,j\neq \mu\\
		0, \quad \text{otherwise}
\end{array}.
\right.
\end{eqnarray}

As mentioned before, fixing different values of the reference mode $\mu$, generates distinct maps $\Phi_t^{\mu}$  with
domain  $\mathcal{S}_{r(2)}^{\mu,pure}$. 
Now let us compare these  distinct maps.  
We know that given two sets
$\mathcal{S}_2^{\mu,pure}$ and $\mathcal{S}_2^{\nu,pure}$, with fixed modes $\mu$ and $\nu$,  
there exists a unitary $V\in\mathcal{U}(\mathcal{F}_{2}^{L+1})$ such that $a_{\nu}^{\dagger}a_{k}^{\dagger}\ket{0} =V a_{\mu}^{\dagger}a_{k}^{\dagger}\ket{0}$. 
Therefore, any pair of maps $\Phi_t^{\mu}$ and $\Phi_t^{\nu}$ have the Kraus 
operators $\{K_{j}^\mu=f_{j}U_ta_{\mu}^{\dagger}\}_j$ and $\{E_{j}^\nu=f_{j}U_tVa_{\mu}^{\dagger}\}_j$, respectively. We 
can compute an upper bound to the norm difference of 
the (Choi-Jamiolkowski) dynamical matrices $D_{\Phi_t^{\mu}}$ and $D_{\Phi_t^{\nu}}$,  associated with the maps, which 
is proved in Appendix \ref{boundip} :
\begin{align}
&\N{D_{\Phi_t^{\mu}} - D_{\Phi_t^{\nu}}}_1 \leq \nonumber\\
& d^{2}L^{2}\sup_{a_{\vec{k}}^{\dagger}\ketbra{0}{0}a_{\vec{k}^{\prime}}\in \mathcal{F}_{2}^{L+1}}\N{\left(a_{\vec{k}}^{\dagger}\ketbra{0}{0}a_{\vec{k}^{\prime}}-V^{T}a_{\vec{k}}^{\dagger}\ketbra{0}{0}a_{\vec{k}^{\prime}}V^{*}\right)}_1,
\end{align}
where $d$ is the dimension of  $\mathcal{F}_{2}^{L+1}$.
It is illustrative to compare this bound with its counterpart in the case of distinguishable particles, where we have 
initially uncorrelated system $S$ and environment $E$  forming a closed global system,  whose dynamics is described by a unitary $U_{S:E}$. 
 Assuming two dynamical maps, $\Phi_t$ and $\Lambda_t$,  constructed from different initial states of the environment, we have the two sets of  Kraus operators $\{ K_a = \bra{a} U_{S:E}\ket{0} \}_a$ and $\{ E_a = \bra{a} U_{S:E}(\mathbb{I}_{S} \otimes V_E)\ket{0} \}_a$, respectively. Then the following inequality, which is proved in Appendix  \ref{bound}, holds: 
\begin{equation}
\N{D_{\Phi} - D_{\Lambda}}_{1} \leq  d_{S}^2\N{\ketbra{0}{0} - V_E \ketbra{0}{0}V_E^{\dagger} }_1,
\end{equation}
where $d_{S}$ is the dimension of the Hilbert space of the system S.
 It is important to emphasize that the two frameworks are completely different. 
 A tensor product structure between system and environment is absent in our context of indistinguishable fermions. 
 Another remark is that the two maps in the distinguishable particles case have the same domain,  which in general is not true in the case of indistinguishable fermions.

\section{Examples of  One-Particle  Dynamical Maps  of Indistinguishable Fermions}
In this section we illustrate our formalism, deriving  the Kraus operators for  the dynamics of one-fermion reduced state of two distinct two-particle Hamiltonians.
 To simplify the discussion,  we assume initial pure global state, such that  
 the Kraus operators $\{K_{l}^\mu =f_{l}U_ta_{\mu}^{\dagger}\}$ have domain given by  Eq.(\ref{rho0}).

\subsection{Non-interacting  Hamiltonian}

Our first example, consisting of a non-interacting  Hamiltonian, shows the consistency of our formalism. As no correlation can be created, and the initial global state is pure, 
it is expected the one-particle evolution be unitary.  The  Hamiltonian   
can be written  in terms of fermionic operators as   $H=\sum_{i,j}{M_{i,j}a_{i}^{\dagger}a_{j}}$, and has the following diagonal form: $H=\sum_{k}{\lambda_{k}b_{k}^{\dagger}b_{k}}$,   where
\begin{eqnarray}
b_{k}^{\dagger}&=&\sum_{i}V_{k,i}a_{i}^{\dagger},   \\
a_{j}^{\dagger}&=&\sum_{k}V_{k,j}^{*}b_{k}^{\dagger}, 
\end{eqnarray}
$\lambda_k$ are the single particle energy excitations and  $V$ is the unitary that diagonalizes $M$.
 The dynamical evolution is given by the unitary $U_{t}=exp{\left(-it\sum_{k}{\lambda_{k}b_{k}^{\dagger}b_{k}}\right)}$. 
 Now, we form  the Kraus operators using Eq.(\ref{kraus2p}),  with the choice 
 $\{f_{k}^{\dagger}\}_{k=0}^{L}=\{b_{k}\}_{k=0}^L$, namely:  $K_{l}^\mu=b_{l}U_{t}a_{\mu}$. 
The matrix elements of the Kraus operator are explicitly:
\begin{eqnarray}
\left(K_{l}^\mu\right)_{m,n}&=&\bra{0}b_{m}b_{l}U_{t}a_{\mu}^{\dagger}b_{n}^{\dagger}\ket{0}\nonumber\\
&=&\bra{0}b_{m}b_{l}U_{t}\left(\sum_{k}V_{k,\mu}^{*}b_{k}^{\dagger}\right)b_{n}^{\dagger}\ket{0}\nonumber\\
&=&\sum_{k}V_{k,\mu}^{*}e^{-it(\lambda_{k}+\lambda_{n})}\left(\delta_{l,k}\delta_{m,n}-\delta_{m,k}\delta_{l,n}\right),\nonumber\\
\end{eqnarray}
thus
\begin{equation}\label{kraus2}
K_{l}^\mu=\sum_{m}e^{-it(\lambda_{l}+\lambda_{m})}\left(V_{l,\mu}^{*}b_{m}^{\dagger}\ketbra{0}{0}b_{m}-V_{m,\mu}^{*}b_{m}^{\dagger}\ketbra{0}{0}b_{l}\right).\nonumber\\
\end{equation}
The map acts on  its domain (Eq.(\ref{rho0}))  as the unitary $U_{t}$:
\begin{eqnarray}
\rho_{r}(t)&=&\frac{1}{2}\sum_{m,n}\left(V_{m,k}^{*}V_{n,k}+V_{m,\mu}^{*}V_{n,\mu}\right)\nonumber\\
&&\times e^{-it\left(\lambda_{m}-\lambda_{n}\right)}b_{m}^{\dagger}\ketbra{0}{0}b_{n}\nonumber\\
&=&U_{t}\rho_{r}(0){U_{t}}^{\dagger}.
\end{eqnarray}

\subsection{Four Level Interacting System}

Consider two spin-$1/2$ fermions, in a lattice of two sites, whose dynamics is given by the following Hamiltonian:
\begin{equation}\label{EHM}
H=-{\sum_{\sigma=\uparrow\downarrow}{\left(a_{1\sigma}^{\dagger}a_{2\sigma}+h.c\right)}}+u\sum_{j=1}^{2}{n_{j\uparrow}n_{j\downarrow}}+v{n_{1}n_{2}},
\end{equation}
where $a_{j\sigma}^{\dagger}$ and $a_{j\sigma}$ are creation and annihilation operators, respectively, of a fermion 
at site $j$ with spin $\sigma$,
 $n_{j\sigma} = a_{j\sigma}^{\dagger}a_{j\sigma}$ and   $n_j = n_{j\uparrow} +n_{j\downarrow}$ are the number operators. 
The first term of the Hamiltonian characterizes  hopping (tunnelling)  between  sites, while the second and third terms characterize the on-site and intersite interactions, parametrized by $u$ and $v$, respectively. 
In the basis $a_{\vec{k}}^{\dagger}\ket{0}\in\mathcal{F}_{2}^{4}$ , where $\vec{k} = (k_1, k_2)$ has  six possible configurations,
\begin{eqnarray}
\vec{k}&\in&\left\{\left(1\!\uparrow,1\!\downarrow\right),\left(1\!\uparrow,2\!\uparrow\right),\left(1\!\uparrow,2\!\downarrow\right),\left(1\!\downarrow,2\!\uparrow\right),\right.\nonumber\\
&&\left.\left(1\!\downarrow,2\!\uparrow\right),\left(2\!\uparrow,2\!\downarrow\right)\right\},
\end{eqnarray}
we obtain the following matrix representation for the Hamiltonian:
\begin{equation}\label{EHM2}
H = \left(\begin{matrix}
		u & 0 & -1 & 1 & 0 & 0\\
		0 & v & 0 & 0 & 0 & 0 \\
		-1 & 0 & v & 0 & 0 & -1 \\
                1 & 0 & 0 & v & 0 & 1 \\
		0 & v & 0 & 0 & v & 0 \\
                0 & 0 & -1 & 1 & 0 & u \\
\end{matrix}\right).
\end{equation}

Now we form the Kraus operators $K_{j}^\mu=a_{j}U_{t}a_{\mu}^{\dagger}$, with the choice $\{f_{k}^{\dagger}\}_{k=0}^{L}=\{a_{k}\}_{k=0}^L$. If the unitary  $V$ diagonalizes the Hamiltonian,  $D=VHV^{\dagger}$, we can write $U_t$ as:
\begin{equation}
U_{t}=\sum_{\vec{l}}e^{-iD_{\vec{l},\vec{l}}t}\sum_{\vec{k},\vec{k^{\prime}}}V_{\vec{l},\vec{k}}V_{\vec{l},\vec{k^{\prime}}}^{*}a_{\vec{k}}^{\dagger}\ketbra{0}{0}a_{\vec{k^{\prime}}}.
\end{equation}
According to  Eq.\ref{kraus2p} we have:
\begin{eqnarray}\label{kraushb}
K_{j}^\mu&=&a_{j}U_t a^\dagger_{\mu}\nonumber\\
&=&\sum_{\vec{l}}e^{-iD_{\vec{l},\vec{l}}t}\sum_{k_{1},k_{2},k_{1}^{\prime},k_{2}^{\prime}}V_{\vec{l},k_{1}k_{2}}V_{\vec{l},k_{1}^{\prime}k_{2}^{\prime}}^{*}\times\nonumber\\
&&a_{j}a_{k_1}^{\dagger}a_{k_2}^{\dagger}\ketbra{0}{0}a_{k^{\prime}_2}a_{k^{\prime}_1}a_{\mu}^{\dagger}.
\end{eqnarray}
Using the  anti-commutation  relations, the last line of Eq.(\ref{kraushb}) reduces to:
\begin{align}
&a_{j}a_{k_1}^{\dagger}a_{k_2}^{\dagger}\ketbra{0}{0}a_{k^{\prime}_2}a_{k^{\prime}_1}a_{\mu}^{\dagger}=\nonumber\\
&=\left(\delta_{j,k_1}a_{k_2}^{\dagger}-\delta_{j,k_2}a_{k_1}^{\dagger}\right)\ketbra{0}{0}\left(a_{k^{\prime}_2}\delta_{k_{1}^{\prime},\mu}-a_{k^{\prime}_1}\delta_{k_{2}^{\prime},\mu}\right),
\end{align}
and finally,
\begin{eqnarray}
K_{j}^\mu &=&\sum_{\vec{l}}e^{-iD_{\vec{l},\vec{l}}t}\sum_{k,k^{\prime}}\left[V_{\vec{l},jk}\left(V_{\vec{l},k^{\prime}\mu}^{*}-V_{\vec{l},\mu k^{\prime}}^{*}\right)+\right.\nonumber\\
&&\left.V_{\vec{l},kj}\left(V_{\vec{l},kj}^{*}-V_{\vec{l},\mu k^{\prime}}^{*}\right)\right]a_{k}^{\dagger}\ketbra{0}{0}a_{k^{\prime}}.
\end{eqnarray}
The unitary V can now be written explicitly as,
\begin{equation}
V = \left(\begin{matrix}
		-\frac{1}{\sqrt{2}} & 0 & 0 & 0 & 0 & \frac{1}{\sqrt{2}}\\
		0 & 0 & 0 & 0 & 1 & 0 \\
		0 & 0 & \frac{1}{\sqrt{2}} & \frac{1}{\sqrt{2}} & 0 & 0 \\
                0 & 1 & 0 & 0 & 0 & 0\\
		a(u,v) & 0 & b(u,v) & -b(u,v) & 0 & a(u,v) \\
               b(u,v) & 0 & -a(u,v) & a(u,v) & 0 & b(u,v) \\
\end{matrix}\right),
\end{equation}
while the explicit form of $D$ is:
\begin{eqnarray}
D&=& \rm{diag} \left(u,v,v,v,\frac{1}{2}\left[(u+v)-\sqrt{{\Delta(u,v)}^2+16}\right],\right.\nonumber\\
&&\left.\frac{1}{2}\left[(u+v)+\sqrt{{\Delta(u,v)}^2+16}\right]\right),
\end{eqnarray}
with $\Delta(u,v)=v-u$,
\begin{equation}
a(u,v)=\frac{\Delta(u,v)+\sqrt{{\Delta(u,v)}^2+16}}{\sqrt{2\left[\left(\Delta(u,v)+\sqrt{{\Delta(u,v)}^2+16}\right)^2+16\right]}},\nonumber
\end{equation}
and
\begin{equation}
b(u,v)=\frac{4}{\sqrt{2\left[\left(\Delta(u,v)+\sqrt{{\Delta(u,v)}^2+16}\right)^2+16\right]}}\nonumber.
\end{equation}

\section{Conclusion}

In systems of indistinguishable fermions, antisymmetrization  eliminates the notion of separability, and the very concept of correlation, 
which is an important ingredient in obtaining CP maps for open systems, becomes subtle.
We showed that it is possible to write a CP map for a single fermion, which is part of a system on $N$ indistinguishable particles, for sets of initial global states with no quantumness of correlation.
 We also illustrated our formalism with  examples of  CP maps  corresponding to a non-interacting and an interacting Hamiltonian of two fermions. 
 The extension of our formalism to subsystems with more than one indistinguishable particle, and for the case of bosons presents no difficulty. 
 As many properties of many-body Hamiltonians can be inferred from the single particle reduced state, an interesting investigation would be if 
 any computational gain can be obtained by the employment of the formalism  developed in this article.
 
\acknowledgments 
We acknowledge financial support by the Brazilian agencies INCT-IQ (National Institute of Science and Technology for Quantum Information), FAPEMIG, and CNPq. 

\appendix
\renewcommand{\thesubsection}{\Alph{section}.\arabic{subsection}}

\section{ Dynamical Map for Single-Fermion Reduced State - General Case with Initial Mixed States }\label{sfm}
\subsection{System of  Two Fermions }

Consider a set of mixed quantum states in the antisymmetric space of $L+1$ modes and two fermions, $\rho(0)\in \mathcal{F}_{2}^{L+1}$, written in a basis of Slater determinants:
\begin{align}\label{rho2}
&\mathcal{S}_{2}^{p}=\nonumber\\
&\left\{\rho(0)=\sum_{\mu\in\Sigma,k\in\Gamma}p(\mu)q(k)a_{\mu}^{\dagger}a_{k}^{\dagger}\ketbra{0}{0}a_{\mu}a_{k}\bigm| p\,\,\text{fixed}\right\},
\end{align}
with  both  $\Sigma$ and $\Gamma$  finite, and disjoint, $\Sigma\cap\Gamma=\emptyset$. Let  $\left|\Sigma\right|=d$, $\left|\Gamma\right|=L-d$,   and  $\mathbb{Z}_{L+1}=\{0,1,\ldots,L\}$. 
We took  the $d$ elements of $\Sigma$ from $\mathbb{Z}_{L+1}$,  and the set  $\Gamma$ as $\mathbb{Z}_d\setminus\Sigma$.  
Tracing out one fermion from $\mathcal{S}_{2}^{p}$, we obtain the  single-fermion reduced states, $\{\rho_{r}(0)\}$: 
\begin{align}\label{rho02}
&\mathcal{S}_{r(2)}^{p}=\nonumber\\
&\left\{\rho_{r}(0)=\frac{1}{2}\sum_{k\in\Gamma}{q(k)a_{k}^{\dagger}\ketbra{0}{0}a_{k}}+\frac{1}{2}\sum_{\mu\in\Sigma}p(\mu)a_{\mu}^{\dagger}\ketbra{0}{0}a_{\mu}\bigm|\right.\nonumber\\
&\left.\quad p\,\,\text{fixed}\right\}.
\end{align}

\begin{definition}\label{def2}
A CP map $\Phi_t^{p}$,  describing the dynamics of the single particle reduced state $\rho_{r}(0)\in\mathcal{S}_{r(2)}^{p}$, 
can be written in Kraus representation as:
\begin{equation}
\Phi_t^{p}[\mathord{ \rho_{r}(0)   }]=\sum_{j=0}^{L}\sum_{\mu\in\Sigma}K_{j,\mu}^{p}\mathord{\rho_{r}(0)}K_{j,\mu}^{p\dagger},
\end{equation}
with the Kraus operators: 
\begin{equation}
K_{l,\mu}^p=f_{l}U_{t}{a_\mu}^{\dagger}\sqrt{p(\mu)}\prod_{m\in\Sigma}\left(1-a_{m}^{\dagger}a_{m}\right),
\end{equation}
\end{definition}

\begin{proof}
 The one-particle reduced dynamics can be expressed as $\rho_{r}(t)= Tr_{1}(U_t\rho(0)U_{t}^{\dagger})$:
\begin{align}
&\rho_{r}(t)=\nonumber\\
&=\frac{1}{2}\sum_{k=0}^{L}{f_{l} U_{t}\left(\sum_{\mu\in\Sigma,k\in\Gamma}p(\mu)q(k)a_{\mu}^{\dagger}a_{k}^{\dagger}\ketbra{0}{0}a_{\mu}a_{k}\right)U_{t}^{\dagger}f_{l}^{\dagger}}\nonumber\\\nonumber\\
&=\sum_{l=0}^{L}\sum_{\mu\in\Sigma}\sqrt{p(\mu)}f_{l} U_{t}a_{\mu}^{\dagger}\left(\frac{1}{2}\sum_{k\in\Gamma}{q(k)a_{k}^{\dagger}\ketbra{0}{0}a_{k}}\right)\times\nonumber\\
&\sqrt{p(\mu)}a_{\mu}U_{t}^{\dagger}f_{l}^{\dagger}.
\end{align}
Defining an operator $\prod_{m\in\Sigma}\left(1-a_{m}^{\dagger}a_{m}\right)$ that annihilates  fermions in  $\Sigma$, and leaves states unchanged otherwise, 
we can write
\begin{align}\label{r1t}
&\rho_{r}(t)=\nonumber\\
&=\sum_{l=0}^{L}\sum_{\mu\in\Sigma}\sqrt{p(\mu)}f_{l} U_{t}a_{\mu}^{\dagger}\prod_{m\in\Sigma}\left(1-a_{m}^{\dagger}a_{m}\right)\times\nonumber\\
&\left(\frac{1}{2}\sum_{k\in\Gamma}{q(k)a_{k}^{\dagger}\ketbra{0}{0}a_{k}}\right)\prod_{m\in\Sigma}\left(1-a_{m}^{\dagger}a_{m}\right)a_{\mu}U_{t}^{\dagger}f_{l}^{\dagger}\sqrt{p(\mu)}.\nonumber\\
\end{align}
Note that  
\begin{equation}\label{r0t}
\prod_{m\in\Sigma}\left(1-a_{m}^{\dagger}a_{m}\right)\left(\frac{1}{2}\sum_{j\in\Sigma}{p(j)a_{j}^{\dagger}\ketbra{0}{0}a_{j}}\right)=0.
\end{equation}
Adding Eq.(\ref{r0t}) to Eq.(\ref{r1t}), Definition \ref{def2} is proven:
\begin{eqnarray}
\rho_{r}(t)&=&\sum_{l=0}^{L}\sum_{\mu\in\Sigma}f_{l} U_{t}a_{\mu}^{\dagger}\sqrt{p(\mu)}\prod_{m\in\Sigma}\left(1-a_{m}^{\dagger}a_{m}\right)\nonumber\\
&&\times\left(\frac{1}{2}\sum_{k\in\Gamma}{q(k)a_{k}^{\dagger}\ketbra{0}{0}a_{k}}+\frac{1}{2}\sum_{j\in\Sigma}p(j)a_{j}^{\dagger}\ketbra{0}{0}a_{j}\right)\nonumber\\\nonumber\\
&&\times\prod_{m\in\Sigma}\left(1-a_{m}^{\dagger}a_{m}\right)\sqrt{p(\mu)}a_{\mu}U_{t}^{\dagger}f_{l}^{\dagger}\nonumber\\
&=&\sum_{l=0}^{L}\sum_{\mu\in\Sigma}K_{l,\mu}^p\rho_{r}(0)K_{l,\mu}^{\dagger p}.
\end{eqnarray}
\end{proof}

\subsection{System of $N$-Fermions}

Consider  a set of states  $\rho\in\M{F}_N^{L+1}$, with no quantumness,
\begin{align}\label{rhoN}
&\mathcal{S}_{N}^{p}=\nonumber\\
&\left\{\rho(0) = \sum_{\vec{\mu}\in\vec{\Sigma}} \sum_{k\in\Gamma}p (\mu_1,\cdots,\mu_{N-1})q(k)\right.\nonumber\\
&\left.\quad \times a_{\vec{\mu}}a_{k} \ketbra{0}{0} a_{k}a_{\vec{\mu}}\bigm| p\,\,\text{fixed}\right\},  
\end{align}
where $\vec{\mu} = (\mu_1,\ldots, \mu_{N-1})$,  $\vec{\Sigma} = (\Sigma_1,\ldots, \Sigma_{N-1})$ are $N-1$-tuples, and
 $p(\vec{\mu})$, $q(k)$ are probability distributions.
 The sets $\Sigma_{j}$ and $\Gamma$ are finite, and disjoint $\Sigma_{j}\cap\Gamma=\emptyset$ $\forall {j}$.
With   $\left|\vec{\Sigma}=\cup_{i=1}^{N-1}\Sigma_{i}\right|=d$, $\left|\Gamma\right|=L-d$, and $\mathbb{Z}_{L+1}=\{0,1,\ldots,L\}$,  we took the  $d$ elements of $\cup_{i=1}^{N-1}\Sigma_{i}$
from $\mathbb{Z}_{L+1}$,  and the set $\Gamma$ as $\mathbb{Z}_d\setminus\cup_{i=1}^{N-1}\Sigma_{i}$. 
 Note that $d$ is the number of accessible modes for $N-1$ fermions, thus $d\geq N-1$. 

Tracing  $N-1$ fermions out   from  (\ref{rhoN}), we obtain the set of single-fermion reduced states $\{\rho_{r}(0)\}$:   
\begin{align}
&\mathcal{S}_{r(N)}^p=\left\{\rho_{r}(0)=\frac{1}{N}\sum_{k\in\Gamma}{q(k)a_{k}^{\dagger}\ketbra{0}{0}a_{k}}+\right.\nonumber\\
&\left.\quad\frac{1}{N}\sum_{j=1}^{N-1}\sum_{\mu_{j}\in\Sigma_{j}}p_j(\mu_j)a_{\mu_{j}}^{\dagger}\ketbra{0}{0}a_{\mu_{j}}\bigm|\right.\nonumber\\
&\left.\quad 	{p}_{j}\,\,\text{fixed}\,\forall{j}\right\},
\end{align}
where $p_j(\mu_j)=\sum_{\vec{\mu}\setminus\mu_{j}}p (\mu_1,\cdots,\mu_{N-1})$ is the marginal distribution.

\begin{definition}\label{defN}
A CP map $\Phi_t^{p}$  describing  the dynamics of the  single particle reduced state $\rho_{r}(0)\in\mathcal{S}_{r(N)}^{p}$,
 can be written in Kraus representation as:
\begin{equation}
\Phi_t^p[\mathord{  \rho_{r}(0) }]=\sum_{\vec{l},\vec{\mu}}^{L}K_{\vec{l},\vec{\mu}}^{p}  \rho_{r}(0) K_{\vec{l},\vec{\mu}}^{p\dagger },
\end{equation}
with  the Kraus operators:  
\begin{align}
& K_{\vec{l},\vec{\mu}}^{p}=\nonumber\\
&=\sqrt{p (\mu_1,\cdots,\mu_{N-1})}\times\nonumber\\
& \quad f_{\vec{l}}Ua_{\vec{\mu}}^{\dagger}\prod_{m\in\cup_{i=1}^{N-1} \Sigma_{i}}\left(1-a_{m}^{\dagger}a_{m}\right).
\end{align}
\end{definition}

The proof of Definition \ref{defN} is \textit{mutatis mutandis} the same performed for Definition \ref{def2}.

\section{Norm Bound } 

\subsection{ Fermionic System}
\label{boundip}
\begin{theorem}
Consider two maps $\Phi$ and $\Lambda$, 
 with Kraus operators $K_{j}=f_{j}Ua_{\mu}$ and $E_{j}=f_{j}UVa_{\mu}$, respectively.
Then the following inequality holds: 
\begin{align}
&\N{D_{\Phi} - D_{\Lambda}}_1 \leq \nonumber\\
& d^{2}L^{2}\sup_{a_{\vec{k}}^{\dagger}\ketbra{0}{0}a_{\vec{k}^{\prime}}\in \mathcal{F}_{2}^{L+1}}\N{\left(a_{\vec{k}}^{\dagger}\ketbra{0}{0}a_{\vec{k}^{\prime}}-V^{T}a_{\vec{k}}^{\dagger}\ketbra{0}{0}a_{\vec{k}^{\prime}}V^{*}\right)}_1,
\end{align}
where $d$ is the dimension of  $\mathcal{F}_{2}^{L+1}$, $\vec{k}=(k_1,k_2)$ is a $2$-tuple indicating the modes occupied by a pair of fermions,
with $k_i=0,\cdots,L$, and $V$ is a unitary operator, $V:\mathcal{F}_{2}^{L+1}\mapsto\mathcal{F}_{2}^{L+1}$.
\end{theorem} 
\begin{proof}
Writing the dynamical matrix of a map $\Phi$ in terms of the Kraus operators $\{K_{j}\}$:
\begin{equation}
D_{\Phi} = \sum_{j}\text{vec}\!\left(K_{j}\right)\text{vec}\!\left(K_{j}\right)^{\dagger}, 
\end{equation}
where the {vec} operation is defined by $\text{vec}\!\left(\ketbra{x}{y}\right)= \ket{x}\otimes\ket{y}$,  we obtain:

\begin{align}
&\N{D_{\Phi} - D_{\Lambda}}_1 =\nonumber\\
&= \N{ \sum_{j}\left(\text{vec}\!\left(K_{j}\right)\text{vec}\!\left(K_{j}\right)^{\dagger}-\text{vec}\!\left(E_{j}\right)\text{vec}\!\left(E_{j}\right)^{\dagger}\right)}_1\nonumber\\
&= \N{ \sum_{j}\left(\text{vec}\!\left(f_{j}Ua_{\mu}\right)\text{vec}\!\left(f_{j}Ua_{\mu}\right)^{\dagger}-\right.\nonumber\\
&\quad\left.\text{vec}\!\left(a_{j}UVa_{\mu}\right)\text{vec}\!\left(f_{j}UVa_{\mu}\right)^{\dagger}\right)}_1.
\end{align} 

Using the following identity for matrices:
\begin{equation}
\text{vec}\!\left(ABC\right)=\left(A\otimes C^{T}\right) \text{vec}\!\left(B\right),
\end{equation}
we have,
\begin{align}\label{norm1}
&\N{D_{\Phi} - D_{\Lambda}}_1 =\nonumber\\
&= \N{ \sum_{j}\left(f_{j}\otimes a_{\mu}^{*}\text{vec}\!\left(U\right)\text{vec}\!\left(U\right)^{\dagger}f_{j}^{\dagger}\otimes a_{\mu}^{T}-\right.\nonumber\\
&\quad\left. f_{j}\otimes a_{\mu}^{*}V^{T}\text{vec}\!\left(U\right)\text{vec}\!\left(U\right)^{\dagger}f_{j}^{\dagger}\otimes V^{*}a_{\mu}^{T}\right)}_1.
\end{align} 

With the unitary operator $U$ written as,
\begin{equation}
U=\sum_{\vec{k},\vec{k^{\prime}}}u_{\vec{k},\vec{k^{\prime}}}a_{\vec{k}}^{\dagger}\ketbra{0}{0}a_{\vec{k}^{\prime}}
\end{equation}
where $\vec{k} = (k_1, k_2)$,   Eq.\eqref{norm1} becomes:
\begin{align}
&\N{D_{\Phi} - D_{\Lambda}}_1 =\nonumber\\
&= \N{ \sum_{j}\sum_{\vec{k},\vec{k^{\prime}}\vec{l},\vec{l^{\prime}}}u_{\vec{k},\vec{k^{\prime}}}u_{\vec{l},\vec{l^{\prime}}}^{*}\left[\left(f_{j}a_{\vec{k}}^{\dagger}\ket{0}\otimes a_{\mu}^{*}a_{\vec{k}^{\prime}}^{\dagger}\ket{0}\right)\times\right.\nonumber\\
&\left.\left(\bra{0}a_{\vec{l}}f_{j}^{\dagger}\otimes \bra{0}a_{\vec{l}^{\prime}}a_{\mu}^{T}\right)-\left(f_{j}a_{\vec{k}}^{\dagger}\ket{0}\otimes a_{\mu}^{*}V^{T}a_{\vec{k}^{\prime}}^{\dagger}\ket{0}\right)\times\right.\nonumber\\
&\left.\left(\bra{0}a_{\vec{l}}f_{j}^{\dagger}\otimes \bra{0}a_{\vec{l}^{\prime}}V^{*}a_{\mu}^{T}\right)\right]}_1\nonumber\\
&= \N{\sum_{\vec{k},\vec{k^{\prime}}\vec{l},\vec{l^{\prime}}}u_{\vec{k},\vec{k^{\prime}}}u_{\vec{l},\vec{l^{\prime}}}^{*}\left[\sum_{j}\left(f_{j}a_{\vec{k}}^{\dagger}\ketbra{0}{0}a_{\vec{l}}f_{j}^{\dagger}\right)\otimes\right.\nonumber\\
&\left.\left(a_{\mu}^{*}a_{\vec{k}^{\prime}}^{\dagger}\ketbra{0}{0}a_{\vec{l}^{\prime}}a_{\mu}^{T}-a_{\mu}^{*}V^{T}a_{\vec{k}^{\prime}}^{\dagger}\ketbra{0}{0}a_{\vec{l}^{\prime}}V^{*}a_{\mu}^{T}\right)\right]}_1.
\end{align} 
Using some norm properties, as triangle inequality $(\N{X+Y}\leq\N{X}+\N{Y})$, positive scalability $(\N{\alpha X}=|\alpha|\N{X},\,\,\alpha\in\mathbb{C})$,  and tensor product ($\N{X_{1}\otimes X_{2}}= \N{X_{1}}\N{X_{2}}$) and the definition of  fermionic partial trace of one particle, we can write:

\begin{align}
&\N{D_{\Phi} - D_{\Lambda}}_1 \leq \nonumber\\
& \sum_{\vec{k},\vec{k^{\prime}}\vec{l},\vec{l^{\prime}}}\left|u_{\vec{k},\vec{k^{\prime}}}u_{\vec{l},\vec{l^{\prime}}}^{*}\right|\N{Tr_1\left(a_{\vec{k}}^{\dagger}\ketbra{0}{0}a_{\vec{l}}\right)}_{1}\times\nonumber\\
&\N{a_{\mu}^{*}\left(a_{\vec{k}^{\prime}}^{\dagger}\ketbra{0}{0}a_{\vec{l}^{\prime}}-V^{T}a_{\vec{k}^{\prime}}^{\dagger}\ketbra{0}{0}a_{\vec{l}^{\prime}}V^{*}\right)a_{\mu}^{T}}_1.
\end{align}
As  the trace norm is non-increasing under partial trace $(\N{Tr_{X_{2}}\left(X\right)}_{1} \leq \N{A}_{1})$,  is sub-multiplicative ($\N{XY}_{1}\leq\N{X}_{1}\N{Y}_{1}$), and we also have
 $\N{X}_{1}=\N{X^{\dagger}}_{1}=\N{X^{T}}_{1}=\N{X^{*}}_{1}$:
\begin{align}
&\N{D_{\Phi} - D_{\Lambda}}_1 \leq \\
& \sum_{\vec{k},\vec{k^{\prime}}\vec{l},\vec{l^{\prime}}}\left|u_{\vec{k},\vec{k^{\prime}}}u_{\vec{l},\vec{l^{\prime}}}^{*}\right|\N{a_{\vec{k}}^{\dagger}\ketbra{0}{0}a_{\vec{l}}}_{1}\times\nonumber\\
&\N{\left(a_{\vec{k}^{\prime}}^{\dagger}\ketbra{0}{0}a_{\vec{l}^{\prime}}-V^{T}a_{\vec{k}^{\prime}}^{\dagger}\ketbra{0}{0}a_{\vec{l}^{\prime}}V^{*}\right)}_1\N{a_{\mu}}_{1}^{2}.
\end{align}
As  $\N{a_{\vec{k}}^{\dagger}\ketbra{0}{0}a_{\vec{l}}}_{1}=Tr\sqrt{a_{\vec{l}}^{\dagger}\ketbra{0}{0}a_{\vec{l}}}=1$, and  $\N{a_{\mu}}_{1}=Tr\sqrt{n_{\mu}}=L$ is the number of
 states $\{a_{\vec{k}}^{\dagger}\ket{0}\}$ with occupied mode $\mu$:
\begin{align}
&\N{D_{\Phi} - D_{\Lambda}}_1 \leq \nonumber\\
&L^2 \sum_{\vec{k},\vec{k^{\prime}}\vec{l},\vec{l^{\prime}}}\sqrt{u_{\vec{k},\vec{k^{\prime}}}u_{\vec{l},\vec{l^{\prime}}}^{*}u_{\vec{k},\vec{k^{\prime}}}^{*}u_{\vec{l},\vec{l^{\prime}}}}\times\nonumber\\
&\N{\left(a_{\vec{k}^{\prime}}^{\dagger}\ketbra{0}{0}a_{\vec{l}^{\prime}}-V^{T}a_{\vec{k}^{\prime}}^{\dagger}\ketbra{0}{0}a_{\vec{l}^{\prime}}V^{*}\right)}_1.
\end{align}
From the definition of unitary operators we have, $\sum_{k}u_{i,k}^{*}u_{j,k}=\sum_{k}u_{k,i}^{*}u_{k,j}=\delta_{i,j}$, therefore:
\begin{align}
&\N{D_{\Phi} - D_{\Lambda}}_1 \leq \nonumber\\
&L^2 \sum_{\vec{k^{\prime}}\vec{l^{\prime}}}\N{\left(a_{\vec{k}^{\prime}}^{\dagger}\ketbra{0}{0}a_{\vec{l}^{\prime}}-V^{T}a_{\vec{k}^{\prime}}^{\dagger}\ketbra{0}{0}a_{\vec{l}^{\prime}}V^{*}\right)}_1.
\end{align}
Finally,
\begin{align}
&\N{D_{\Phi} - D_{\Lambda}}_1 \leq \nonumber\\
& d^{2}L^{2}\sup_{a_{\vec{k}}^{\dagger}\ketbra{0}{0}a_{\vec{k}^{\prime}}\in \mathcal{F}_{2}^{L+1}}\N{\left(a_{\vec{k}}^{\dagger}\ketbra{0}{0}a_{\vec{k}^{\prime}}-V^{T}a_{\vec{k}}^{\dagger}\ketbra{0}{0}a_{\vec{k}^{\prime}}V^{*}\right)}_1.
\end{align}
\end{proof}

\subsection{ System of Distinguishable Particles}\label{bound}

\begin{theorem}
Assume two maps $\Phi$ and $\Lambda$, with Kraus operators $\{ K_a = \bra{a} U_{S:E}\ket{0} \}_a$ and  $\{ E_a = \bra{a} U_{S:E}(\mathbb{I}_{S} \otimes V_E)\ket{0} \}_a$, respectively. Then the following inequality holds: 
\begin{equation}
\N{D_{\Phi} - D_{\Lambda}}_{1} \leq  d_{S}^2\N{\ketbra{0}{0} - V_E \ketbra{0}{0}V_E^{\dagger} }_1,
\end{equation}
where $d_{S}$ is the dimension of the Hilbert space of the system $S$.
\end{theorem} 

\begin{proof}
Writing the dynamical matrix of a map $\Phi$ in the Choi representation:
\begin{equation}
D_{\Phi} = \sum_{i,j=1}^{d_{S}} \Phi(\ketbra{i}{j}) \otimes \ketbra{i}{j}, 
\end{equation}
we obtain: 
\begin{align}
&\N{D_{\Phi} - D_{\Lambda}}_1 =\nonumber\\
&= \frac{1}{d_{S}^2} \N{\sum_{i,j=1}^{d_{S}} \Phi(\ketbra{i}{j}) \otimes \ketbra{i}{j} - \sum_{i,j=1}^{d_{S}} \Lambda(\ketbra{i}{j}) \otimes \ketbra{i}{j}}_1\nonumber\\
&= \N{\sum_{i,j=1}^{d_{S}} \left\{ \sum_{a} K_a \ketbra{i}{j} K_{a}^{\dagger} - E_a \ketbra{i}{j} E_{a}^{\dagger}  \right\}\otimes \ketbra{i}{j} }_1 \nonumber\\
&\leq\sum_{i,j=1}^{d_{S}} \N{ \left\{ \sum_{a} K_a \ketbra{i}{j} K_{a}^{\dagger} - E_a \ketbra{i}{j} E_{a}^{\dagger}  \right\}\otimes \ketbra{i}{j} }_1. \label{eq40}
\end{align} 
Thus, by the definition of Kraus operators above: 
\begin{align}
&K_a \ketbra{i}{j} K_{a}^{\dagger} - E_a \ketbra{i}{j} E_{a}^{\dagger} = \nonumber\\
&=\bra{a}_E \left\{ U_{SE}\ketbra{i}{j}_S\otimes (\ketbra{0}{0}_E - V\ketbra{0}{0}_E V^{\dagger}) U_{SE}^{\dagger } \right\} \ket{a}_E,
\end{align}
substituting in Eq.\eqref{eq40}, and using  $\N{X\otimes Y} = \N{X}\N{Y}$:
\begin{align}
&\N{D_{\Phi} - D_{\Lambda}}_1 \leq \nonumber\\
&d_{S}^2\N{ \left\{ \sum_{a} \bra{a} \left[ U_{SE}\ketbra{i}{j}\otimes (\ketbra{0}{0} - V\ketbra{0}{0}V^{\dagger}) U_{SE}^{\dagger } \right] \ket{a} \right\} }_1 \nonumber\\
& = d_{S}^2\N{ \left\{  \left[ U_{SE}\ketbra{i}{j}\otimes (\ketbra{0}{0} - V\ketbra{0}{0}V^{\dagger}) U_{SE}^{\dagger } \right] \right\} }_1,
\end{align}
where we used  that $\sum_a \bra{a} X \ket{a} = \text{Tr}_{E}(X)$. Finally, as trace distance is invariant under unitary operations, the statement is proved:
 \begin{equation}
 \N{D_{\Phi} - D_{\Lambda}}_1 \leq d_{S}^2\N{  
 (\ketbra{0}{0} - V\ketbra{0}{0}V^{\dagger})}_1.
 \end{equation}
\end{proof}


\end{document}